\documentclass[a4paper,twcolumn]{IEEEtran}
\usepackage{color}
\usepackage{graphicx}
\usepackage{amssymb}
\usepackage{multicol}
\usepackage{amsmath}
\usepackage{amsthm}

\usepackage{mathtools}
\newtheorem{pro}{\textit{Proposition}}

\newcommand\blfootnote[1]{%
  \begingroup
  \renewcommand\thefootnote{}\footnote{#1}%
  \addtocounter{footnote}{-1}%
  \endgroup
}

\begin{document}
\title{Subverting Massive MIMO by Smart Jamming}
\author{Hessam Pirzadeh, S. Mohammad Razavizadeh, and Emil Bj\"{o}rnson}
\maketitle
\begin{abstract}
We consider uplink transmission of a  massive multi-user multiple-input multiple-output (MU-MIMO) system in the presence of a smart jammer. The jammer aims to degrade the sum spectral efficiency of the legitimate system by attacking both the training and data transmission phases. First, we derive a closed-form expression for the sum spectral efficiency by taking into account the presence of a smart jammer. Then, we determine how a jammer with a given energy budget should attack the training and data transmission phases to induce the maximum loss to the sum spectral efficiency. Numerical results illustrate the impact of optimal jamming specifically in the large limit of the number of base station (BS) antennas.

\end{abstract}

\vspace{-7mm}
\blfootnote{The work of E. Bj\"{o}rnson was supported by Security Link.

H. Pirzadeh and S. M. Razavizadeh are with the School of Electrical Engineering, Iran University of Science and Technology (IUST), Tehran 1684613114, Iran (e-mail: h.pirzadeh.1990@ieee.org; smrazavi@iust.ac.ir).

E. Bj\"{o}rnson is with the Department of Electrical Engineering (ISY), Link\"{o}ping University, Link\"{o}ping, Sweden (e-mail: emil.bjornson@liu.se).}
\section{Introduction}\label{sec:Introduction}
\IEEEPARstart{J}{amming} constitutes a critical problem for reliability in wireless communications and imposes detrimental impact on the performance of wireless systems. In recent years, the performance of wireless systems in the presence of jamming has been studied extensively from a communication theoretic perspective. The problem of jammer design in a training-based point-to-point MIMO system and a multi-user system with a single antenna BS is analyzed in \cite{Zhou} and \cite{Pezeshki}, respectively. In addition, the design of a full-duplex eavesdropper (jammer) based on pilot contamination attack is considered in \cite{Maham}.

Recently, massive MIMO systems have attracted lots of attention. The ability to increase the spectral efficiency (SE) along with improving the energy efficiency have made the technology one of the main candidates for next generation wireless networks \cite{Ngo1}-\cite{Larsson}.
In spite of the great amount of research regarding massive MU-MIMO systems, there are still only a few works in this area that have adopted physical layer security (PLS) issues into their analyses.
However,
the revisiting of PLS seems to be necessary in massive MIMO systems.
It is recognized that massive MIMO brings new challenges and opportunities in this area which is unique and thoroughly different from conventional MIMO systems \cite{Ozan}, \cite{Kapetanovic1}. Secure transmission in the downlink of a massive MIMO system in the presence of an adversary which is capable of both jamming and eavesdropping is analysed in \cite{Ozan}, the problem of active and passive eavesdropping is studied in \cite{Kapetanovic1}, and secrecy enhancement in the downlink of a massive MIMO system in the presence of eavesdropper(s) is considered in \cite{Zhu}-\cite{Wang}.
Nevertheless, to the best of the authors' knowledge, no study has been done on analyzing the massive MU-MIMO system's performance (in terms of sum SE) in the presence of a jammer.
In this paper, we investigate the design of a smart jammer in the uplink of a single-cell massive MU-MIMO system and study the effect of the jamming on the performance of the system.
To design a smart jammer, by adopting classical bounding techniques, we derive a closed-form expression for the sum SE of the massive MU-MIMO systems in the presence of a smart jammer. Then, we find the optimal strategy that a jammer with a given energy budget should employ to induce the maximum loss to the sum SE of a legitimate massive MU-MIMO system.
Analytical and numerical analyses show that to what extent the abundance of antenna elements at the BS can increase the susceptibility of the massive MU-MIMO systems to the jamming attack. It is also shown that the advantage of optimal jamming over fixed power jamming boosts as the number of BS antennas goes large.

\emph{Notations}: We use boldface to denote matrices and vectors. $(.)^*$, $(.)^T$, $(.)^H$ and $(.)^\star$ denotes conjugate, transpose, conjugate transpose, and optimal value, respectively. $\boldsymbol{v}\sim\mathcal{CN}(\boldsymbol{0},\boldsymbol{\boldsymbol{R}})$ denotes circularly-symmetric complex Gaussian (CSCG) random vector with zero mean and covariance matrix $\mathit{\boldsymbol{R}}$. $\|.\|$ denotes Euclidean norm. $\mathit{\boldsymbol{I}_K}$ is the $K \times K$ identity matrix and expectation operator is denoted by $\mathbb{E}\{.\}$.

\section{System Model}\label{sec:SYSTEM MODEL}
Consider the uplink of a single-cell MU-MIMO system consisting of $K$ legitimate, single-antenna users (hereafter \emph{users}) that send their signals simultaneously to a BS equipped with $M$ antennas. Also there is a jammer which aims to reduce the sum SE of the legitimate system by carefully attacking the training and data transmission phases. Accordingly, the $M\times1$ received signal at the BS is
\begin{equation}\label{channel model}
  \mathit{\boldsymbol{y}}=\sqrt{p}\mathit{\boldsymbol{G}}\mathit{\boldsymbol{x}} + \mathit{\boldsymbol{n}}+\sqrt{q}\mathit{\boldsymbol{g}}_ws,
\end{equation}
where $p$ represents the average transmission power from each user, $\boldsymbol{G}=\boldsymbol{HD}^{\frac{1}{2}}$ is the channel matrix where $\boldsymbol{H}\in\mathbb{C}^{M\times K}$ models fast fading with each element, $h_{mk}$, distributed independently as $\mathcal{CN}\left(0,1\right)$
and $\boldsymbol{D}\in\mathbb{R}^{K\times K}$ is a constant diagonal matrix whose $k$th diagonal element, $\beta_k$, models the geometric attenuation and shadow-fading between the $k$th user and the BS.
$\mathit{\boldsymbol{x}}\in\mathbb{C}^{K\times1}$ is the symbol vector transmitted from users and is drawn from a CSCG codebook which satisfies $\mathbb{E}\{\mathit{\boldsymbol{x}}\mathit{\boldsymbol{x}}^{H}\}=\boldsymbol{I}_K$, and $\mathit{\boldsymbol{n}}\sim\mathcal{CN}\left(\mathbf{0},\boldsymbol{I}_M\right)$ denotes additive CSCG receiver noise at the BS.
In addition, $q$ represents the jammer's average power and $\mathit{\boldsymbol{g}}_w\sim\mathcal{CN}(\boldsymbol{0},\beta_w\boldsymbol{I}_M)$ is the channel vector between the jammer and the BS. Finally, $s$ denotes the jammer's symbol where $\mathbb{E}\left\{|s|^2\right\}=1$.

We consider a block-fading model where each channel remains constant in a coherence interval of length $\mathit{T}$ symbols and changes independently between different intervals. Note that $T$ is a fixed system parameter chosen as the minimum coherence duration of all users and is assumed to be smaller than that of the jammer.
At the beginning of each coherence interval, the users send their $\eta$-tuple mutually orthogonal pilot sequences ($K\leq\eta\leq T$) to the BS for channel estimation. The remaining $T-\eta$ symbols are dedicated to uplink data transmission. The average transmission powers of the users during training and data transmission phases are denoted by $p_t$ and $p_d$, respectively.

In order to analyze the worst-case impact of jamming,
we assume that the jammer is aware of the transmission protocol and can potentially use different powers for jamming the training and data transmission phases \cite{Zhou}, \cite{Karlsson}, which are denoted by $q_t$ and $q_d$, respectively.

\subsection{ Training\ Phase}\
The pilot sequences can be stacked into an $\eta\times K$ matrix $\sqrt{\eta p_t}\mathbf{\Phi}$, where the $k$th column of $\mathbf{\Phi}$, $\boldsymbol{\phi}_k$, is the $k$th user's pilot sequence and $\mathbf{\Phi}^H\mathbf{\Phi}=\boldsymbol{I}_K$\footnote{We assume that the legitimate system changes $\mathbf{\Phi}$ randomly in different coherence intervals and, hence, the jammer is unable to know the users' pilot sequences during training phase \cite{Ozan}.}. Therefore, the $M\times\eta$ received signal at the BS is
\begin{equation}\label{training matrix}
  \boldsymbol{Y}_t=\sqrt{\eta p_t}\boldsymbol{G}\mathbf{\Phi}^T+\boldsymbol{N}+\sqrt{\eta q_t}\mathit{\boldsymbol{g}}_w\boldsymbol{\phi}_w^{T},
\end{equation}
where $\boldsymbol{N}$ is an $M\times\eta$ matrix with i.i.d. $\mathcal{CN}(0,1)$ elements, and $\boldsymbol{\phi}_w$ is the jammer's pilot sequence\footnote{Since the jammer does not know the users' pilot sequences, it chooses a random pilot sequence uniformly distributed over the unit sphere, i.e., $\mathbb{E}\left\{\|\boldsymbol{\phi}_w\|^2\right\}=1$, to contaminate pilot sequences of the users \cite{Ngo1}. As a result $\mathbb{E}\left\{|\boldsymbol{\phi}_w^T\boldsymbol{\phi}_k^{*}|^2\right\}=1/\eta$.}.
The minimum mean squared error (MMSE) estimate \cite{Kay} of $\boldsymbol{G}$ given $\boldsymbol{Y}_t$ is
\begin{equation}\label{channel estimate}
  \hat{\boldsymbol{G}}=\frac{1}{\sqrt{\eta p_t}}\boldsymbol{Y}_t\mathbf{\Phi}^{*}{\left(\boldsymbol{I}_K+\frac{1+q_t\beta_w}{\eta p_t}\boldsymbol{D}^{-1}\right)^{-1}}.
\end{equation}
Define $\boldsymbol{\mathcal{E}}\triangleq\hat{\boldsymbol{G}}-\boldsymbol{G}$. Then we have
\begin{equation}\label{ghat variance}
  \sigma_{\hat{g}_k}^{2}=\frac{\eta p_t\beta_k^2}{\eta p_t\beta_k+q_t\beta_w+1}~\text{and}~ \sigma_{{\varepsilon}_k}^{2}=\frac{\left(1+q_t\beta_w\right)\beta_k}{\eta p_t\beta_k+q_t\beta_w+1}
\end{equation}
where $\sigma_{\hat{g}_k}^{2}$ and $\sigma_{{\varepsilon}_k}^{2}$ are the variances of the independent zero-mean elements in the $k$th column of $\hat{\boldsymbol{G}}$ and $\boldsymbol{\mathbf{\mathcal{E}}}$, respectively.
\subsection{ Data\ Transmission\ Phase}\
In this phase, the users send their data to the BS simultaneously while the jammer is sending its artificial noise signal. The BS selects a linear detection matrix ${\boldsymbol{A}}\in\mathbb{C}^{M\times K}$ as a function of the channel estimate $\hat{\boldsymbol{G}}$. Therefore, the resulted signal at the BS is \cite{Ngo1}
\begin{equation}\label{received signal}
  \mathit{\boldsymbol{r}}={\boldsymbol{A}}^H\left(\sqrt{p_d}{\mathbf{G}}\mathit{\boldsymbol{x}}+\mathit{\boldsymbol{n}}+\sqrt{q_d}\mathit{\boldsymbol{g}}_ws\right).
\end{equation}
The $k$th element of $\mathit{\boldsymbol{r}}$ is
\begin{multline}\label{kth user received signal}
  r_k=\sqrt{p_d}{\mathit{\boldsymbol{a}}}_{k}^{H}{\mathit{\boldsymbol{g}}}_{k}x_k\\
  +\sqrt{p_d}\sum_{i=1,i\ne k}^{K}{{\mathit{\boldsymbol{a}}}_k^H{\mathit{\boldsymbol{g}}}_{i}x_i}
  +{\mathit{\boldsymbol{a}}}_k^H\mathit{\boldsymbol{n}}+\sqrt{q_{d}}{\mathit{\boldsymbol{a}}}_k^H{\mathit{\boldsymbol{g}}}_ws,
\end{multline}
where ${\mathit{\boldsymbol{a}}}_k$ and ${\mathit{\boldsymbol{g}}}_k$ are the $k$th columns of ${\boldsymbol{A}}$ and ${\boldsymbol{G}}$, respectively. The BS treats ${\mathit{\boldsymbol{a}}}_{k}^{H}{\mathit{\boldsymbol{g}}}_{k}$ as the desired channel and the last three terms of (\ref{kth user received signal}) as worst-case Gaussian noise when decoding the signal\footnote{The jammer transmits Gaussian signal during data transmission phase, since it induces the worst-case interference in this phase.}.
Consequently, an ergodic achievable SE at the $k$th user is \cite{Emil Bj}
\begin{equation}\label{SE definition}
  \mathcal{S}_k=\mathcal{C}\left(\mathit{SINR}_k\right),
\end{equation}
where $\mathcal{C}\left(\gamma\right)\triangleq\left(1-\eta/T\right)\mathrm{log}_2\left(1+\gamma\right)$ and $\mathit{SINR}_k$ is the effective signal-to-interference-and-noise ratio at the $k$th user given by (\ref{uplink achievable rate})
\begin{figure*}[!t]
\begin{equation}\label{uplink achievable rate}
\mathit{SINR}_k=\frac{p_d|\mathbb{E}\left\{{\mathit{\boldsymbol{a}}}_{k}^{H}{\mathit{\boldsymbol{g}}}_k\right\}|^{2}}{p_d\sum_{i=1}^{K}{\mathbb{E}\left\{|{\mathit{\boldsymbol{a}}}_{k}^{H}{\mathit{\boldsymbol{g}}}_i|^{2}\right\}}
-p_d|\mathbb{E}\left\{{\mathit{\boldsymbol{a}}}_{k}^{H}{\mathit{\boldsymbol{g}}}_k\right\}|^{2}+\mathbb{E}\left\{\|{\mathit{\boldsymbol{a}}}_{k}\|^2\right\}
+q_{d}\mathbb{E}\left\{|\mathit{\boldsymbol{a}}_{k}^{H}{\mathit{\boldsymbol{g}}}_w|^2\right\}}
\end{equation}
\hrulefill
    \vspace*{4pt}
    \end{figure*}
shown at the top of the next page \cite{Emil Bj}.
\subsection{ Sum\ Spectral\ Efficiency}\
\color{black}In our analyses, we choose the sum SE (in bit/s/Hz) as our objective function which is defined as $\mathcal{S}\triangleq\sum_{k=1}^{K}{\mathcal{S}_k}$ \cite{Emil Bj}.
By using (\ref{ghat variance}) and (\ref{uplink achievable rate}) and assuming maximum ratio combining (MRC) at the BS\footnote{Note that the main results do not rely on the assumption of MRC detector and similar ones occur for zero-forcing (ZF) and MMSE detectors.} (i.e., $\boldsymbol{A}=\hat{\boldsymbol{G}}$) \cite{Ngo1}, a closed-form expression for the sum SE in the presence of a smart jammer can be derived as (\ref{approx MRC}) at the top of the next page.
\begin{figure*}[!t]
\begin{equation}\label{approx MRC}
\mathcal{S}=\sum_{k=1}^{K}{\mathcal{C}\left(\frac{M\eta p_t\beta_k^2}{\left(\eta p_t\beta_k+ q_{t}\beta_w+1\right)\left(\sum_{i=1}^{K}{\beta_i}+\frac{1}{p_d}\right)+
\eta p_t\beta_k^2+\frac{q_{d}}{p_d}\left(\left(M+2\right) q_{t}\beta_w+\eta p_t\beta_k+1\right)\beta_w}\right)}
\end{equation}
\hrulefill
    \vspace*{4pt}
    \end{figure*}
From equation (\ref{approx MRC}) it is apparent that, as the number of BS antennas $M$ goes to infinity, jamming saturates the performance of the legitimate system due to pilot contamination. More precisely
\begin{equation}\label{saturate}
  \mathcal{S}~{\xrightarrow[M\to\infty]{}}~\sum_{k=1}^{K}{\mathcal{C}\left(\eta\frac{p_t}{q_{t}}\frac{p_d}{q_{d}}\left(\frac{\beta_k}{\beta_w}\right)^2\right)}.
\end{equation}
In Section \ref{sec:Simulation}, we show that optimal jamming can accelerate the pace of this saturation notably.

Furthermore, we consider energy constrained transmission in each coherence block for both users and the jammer. This energy constraint could be interpreted as a constraint on the power budget during the coherence interval \cite{Zhou}. The power budget of each user and the jammer are denoted by $\mathcal{P}$ and $\mathcal{Q}$, respectively. Hence, for each user we have
\begin{equation}\label{user budget}
  \eta p_t+(T-\eta)p_d=\mathcal{P}T,
\end{equation}
and for the jammer we have
\begin{equation}\label{jammer budget}
  \eta q_{t}+(T-\eta)q_{d}=\mathcal{Q}T.
\end{equation}
We denote the fraction of the total energy that each user and the jammer allocate to the training phase by $\varphi$ and $\zeta$, respectively \cite{Zhou}. Accordingly, for each user we have
\begin{equation}\label{phi}
  p_t=\frac{\varphi\mathcal{P}T}{\eta}~\text{and}~p_d=\frac{(1-\varphi)\mathcal{P}T}{T-\eta},
\end{equation}
and for the jammer we have
\begin{equation}\label{zeta}
  q_{t}=\frac{\zeta \mathcal{Q}T}{\eta}~\text{and}~q_{d}=\frac{(1-\zeta)\mathcal{Q}T}{T-\eta}.
\end{equation}

In the next section, we derive the optimal value of $\zeta$ that minimizes the sum SE.
\section{Optimal Resource Allocation}\label{sec:Optimal Power Allocation}
In this section, we show how a smart jammer with a given energy budget should attack the training and data transmission phases to subvert the performance of a massive MU-MIMO system. Since the users have a fixed strategy during uplink transmission, we assume that the smart jammer can acquire the value of $\varphi$ and $\mathcal{P}$ to facilitate its design and impose the maximum loss to the sum SE of the legitimate system \cite{Zhou}, \cite{Pezeshki}. \color{black}Thus, the optimization problem for deriving the optimal value of $\zeta$ is
\begin{equation} \label{optimization jammer1}
{\mathfrak{P}} : \begin{dcases*}
\begin{aligned}
& \underset{\zeta}{\text{minimize}}
&& \mathcal{S} \\
& \text{subject to} & & 0\le\zeta\le 1. \\
\end{aligned}
\end{dcases*}
\end{equation}
The next proposition helps us to solve this optimization problem efficiently.
\begin{pro}
${\mathfrak{P}}$ is a convex optimization problem.
\end{pro}
\begin{proof}
By substituting (\ref{zeta}) into (\ref{approx MRC}), ${\mathfrak{P}}$ can be written as
\begin{equation} \label{optimization jammer2}
{\mathfrak{P}} : \begin{dcases*}
\begin{aligned}
& \underset{\zeta}{\text{minimize}}
&& \left(1-\frac{\eta}{T}\right)\sum_{k=1}^{K}{\mathrm{log}_2\left(1+\frac{1}{f_k\left(\zeta\right)}\right)} \\
& \text{subject to} & & 0\le\zeta\le 1. \\
\end{aligned}
\end{dcases*}
\end{equation}
In this formulation, $f_k\left(\zeta\right)\triangleq\frac{\alpha\left(\zeta\right)}{M\eta p_t\beta_k^2}$ where
\begin{multline}\label{gk}
\alpha\left(\zeta\right)\triangleq \left(\eta p_t\beta_k+\frac{\zeta\beta_w\mathcal{Q}T}{\eta}+1\right)\left(\sum_{i=1}^{K}{\beta_i}+\frac{1}{p_d}\right)\\
+\frac{\left(1-\zeta\right)\beta_w\mathcal{Q}T}{\left(T-\eta\right)p_d}\left((M+2)\frac{\zeta\beta_w\mathcal{Q}T}{\eta}+\eta p_t\beta_k+1\right)+\eta p_t\beta_k^2.
\end{multline}
The second derivative of $f_k\left(\zeta\right)$ is equal to
\begin{equation}\label{second derivative g}
  \frac{\partial^2f_k\left(\zeta\right)}{\partial\zeta^2}=-\frac{2(M+2)\beta_w^2\mathcal{Q}^2T^2}{ M\eta^2 p_tp_d\beta_k^2\left(T-\eta\right)}< 0.
\end{equation}
Hence, $f_k\left(\zeta\right)$ is a concave function. Since $\mathrm{log}_2\left(1+\frac{1}{x}\right)$ is convex and non-increasing function, and from the concavity of $f_k\left(\zeta\right)$, we conclude that $\mathrm{log}_2(1+\frac{1}{f_k\left(\zeta\right)})$ is a convex function of $\zeta$ \cite{Boyd}. Since also the summation of convex functions is convex, the proof is complete.
\end{proof}
As a result, we can find the optimal jammer energy allocation ratio $\zeta^{\star}$ by any convex optimization tool. In Section \ref{sec:Simulation} we evaluate $\zeta^{\star}$ numerically for different values of the jammer's power budget and the number of BS antennas.

To get an insight into the impact of the number of BS antennas on $\zeta^{\star}$, we can obtain a closed-form solution for $\mathfrak{P}$ in the special case of $\boldsymbol{D}=\beta\boldsymbol{I}_K$.
Using Lagrangian multiplier method and Karush-Kuhn-Tucker (KKT) conditions \cite{Boyd}, the optimal energy allocation ratio can be derived analytically as
\begin{equation}\label{zeta star}
  \zeta^{\star}=\begin{dcases*}
\begin{aligned}
& 0,  &&    \mathcal{Q}T<-\kappa, \\
& 1,  &&    \mathcal{Q}T<\kappa, \\
& \frac{\kappa+\mathcal{Q}T}{2\mathcal{Q}T}, && \text{otherwise},
\end{aligned}
\end{dcases*}
\end{equation}
where
\[
\kappa=\frac{\left(K\beta\left(1-\varphi\right)\mathcal{P}T+T-\eta\right)
-\eta\left(\beta\varphi\mathcal{P}T+1\right)}{\beta_w\left(M+2\right)}.
\]
This analytical expression demonstrates that the jammer's optimal strategy is dependent on the number of BS antennas.
For instance, as the number of BS antennas grows, $\kappa$ becomes smaller. Consequently, the optimal jamming strategy, even for a low power jammer, falls into the third case, i.e., attacking both phases.
Specifically, as the number of BS antennas goes to infinity, $\kappa$ tends to zero and the optimal energy allocation ratio tends to $\zeta^{\star}=1/2$.
This result seems logical since by attacking both phases, the jammer can amplify its adverse impact on the system's performance proportional to $M$ by the aid of pilot contamination.
It can also be shown from (\ref{zeta star}) that $\zeta^{\star}$ is a continuous and non-increasing function of the users' energy allocation ratio $\varphi$. Hence, the more energy the users devote to the training phase, the more energy the jammer should employ to jam the data transmission phase.
\section{Numerical Results}\label{sec:Simulation}
We consider a cell with radius $r_c=1000~\text{m}$ in which $K=10$ users are uniformly distributed and no user is closer than $r_h=200~\text{m}$ to the BS. We set $\beta_{k}={z_{k}}/{(r_{k}/r_h)^{\nu}}$, where $z_{k}$ is a log-normal random variable with standard deviation $\sigma_{sh}=8~\text{dB}$ that models shadow-fading, $r_{k}$ is the distance between the BS and the $k$th user and $\nu=3.8$ is the decay exponent. Also we set $T=200$. It is assumed that the users have optimized their energy allocation ratio $\varphi$ and training duration $\eta$ to maximize the sum SE as described in \cite{Ngo2}, \cite{babak}.
Note that owing to the normalization of the noise variance to one, $\mathcal{P}$ and $\mathcal{Q}$ are ``normalized'' power and, therefore,  dimensionless. Accordingly, they are measured in dB in the numerical evaluations.

In Fig. \ref{fig1} the smart jamming benefit for different values of the jammer's power budget is depicted. We compare the optimal jamming (i.e., $\zeta^{\star}$) to the equal jamming (i.e., $\zeta=\eta/T$) for a BS with $M=100$ antennas. The sum SE reduction thanks to optimal energy allocation is significant which approves the potential of smart jamming in the massive MIMO systems.

\begin{figure}
\centering
\includegraphics[width=0.3\textwidth]
{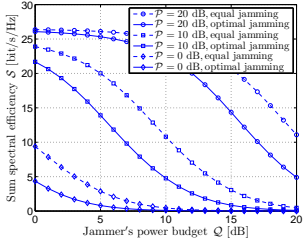}
\caption{Sum spectral efficiency $\mathcal{S}$ versus the jammer's power budget $\mathcal{Q}$ for a BS with $M=100$ antennas. The users' power budget $\mathcal{P}$ is set to 0, 10, and 20 dB. The sum spectral efficiency is depicted when the jammer performs optimal jamming ($\zeta^{\star}$) and equal jamming ($\zeta=\eta/T$).}
\label{fig1}
\end{figure}

Fig. \ref{fig2} illustrates the advantage of smart jamming in the large antenna limit. In our experiment, we set the jammer's power budget to $\mathcal{Q}=10$ dB and the legitimate users' power budget $\mathcal{P}$ to 5, 10, and 15 dB. The divergence between the equal and optimal jamming curves states that the more antennas the BS is equipped, the more harm could be induced to the sum SE by optimal jamming compared to equal jamming. We could also conjecture this phenomenon intuitively. In fact, as the number of BS antennas increases, the adverse impact of the jammer on the sum SE magnifies due to the pilot contamination phenomenon. As a result, the jammer plays the role of a high power jammer which is a scenario that optimal jamming outperforms equal jamming significantly. Moreover, It is evident that optimal jamming saturates the sum SE faster than equal jamming as the number of BS antennas goes large.

\begin{figure}
\centering
\includegraphics[width=0.3\textwidth]
{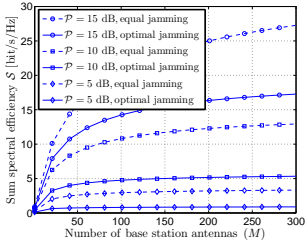}
\caption{Sum spectral efficiency $\mathcal{S}$ versus the number of BS antennas $M$. The jammer's power budget $\mathcal{Q}$ is set to be 10 dB, the users's power budget $\mathcal{P}$ is set to 5, 10, and 15 dB.}
\label{fig2}
\end{figure}

Fig. \ref{fig3} shows the impact of the number of BS antennas $M$ and the users' energy allocation ratio $\varphi$ on the jammer's optimal energy allocation ratio $\zeta^{\star}$.  It demonstrates that, when the number of BS antennas is small, the jammer's optimal energy allocation ratio decreases as the users' energy allocation ratio increases.
For instance, when $M=10$ and $\varphi=0.1$, the jammer should allocate more energy for the jamming of the training phase. The reverse holds true for $\varphi=0.9$, i.e., the jammer should allocate more energy to jam of the data transmission phase. Apart from the dependence of $\zeta^{\star}$ on $\varphi$, the optimal strategy is equal energy allocation (i.e., $\zeta^{\star}=0.5$) as the number of BS antennas goes large. These observations are consistent with the analytical results in Section \ref{sec:Optimal Power Allocation}.

\begin{figure}
\centering
\includegraphics[width=0.3\textwidth]
{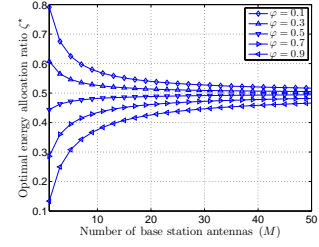}
\caption{Optimal energy allocation ratio $\zeta^{\star}$ versus the number of BS antennas $M$ for different values of $\varphi$. The jammer's power budget $\mathcal{Q}$ and the users's power budget $\mathcal{P}$ is set to be 10 dB.}
\label{fig3}
\end{figure}
\section{Conclusion}\label{sec:concolusion}
In this work, we considered the problem of smart jamming in the uplink of a massive MU-MIMO system.
We showed that if a jammer causes pilot contamination during training phase and optimally allocates its power budget to jam the training and data transmission phases,
it can impose dramatic harm to the sum SE of the legitimate system. Analytical results showed that the optimal strategy of a smart jammer is highly dependent on the number of BS antennas. In particular, when the BS is equipped with large number of antenna elements, even a low power jammer can achieve a large gain by optimal energy allocation over fixed power jamming.

\end{document}